\documentclass[11pt]{article}

\usepackage{amsmath,amsthm,amsfonts,amscd,eucal,latexsym,amssymb,mathrsfs}
\usepackage{amsxtra, calc, bbm}
\usepackage{epsfig}  
\usepackage[latin1]{inputenc}
\usepackage{verbatim} 
\usepackage{graphicx, float}
\usepackage{hyperref}
\usepackage{color}
\usepackage{xcolor}
\usepackage[normalem]{ulem}
\def\cA{{\ca A}}

\def\cE{{\ca E}}
\def\cF{{\ca F}}

\def\cS{{\ca S}}

\def\cK{{\ca K}}

\def\bC{{\mathbb C}}           

\def\bN{{\mathbb N}}

\def\bR{{\mathbb R}}

\def\bS{{\mathbb S}}

\def\beq{\begin{eqnarray}}
\def\eeq{\end{eqnarray}}

\newcommand{\ca}[1]{{\cal #1}}         

\def\alp{\alpha}

\def\de{\delta}

\def\ka{\kappa}
\def\la{\lambda}

\def\om{\omega}

\def\De{\Delta}

\newcommand{\id}{\mathbbm{1}} 

\def\supp{\textrm{supp}}

\newtheorem{theorem}{Theorem}[section]
\newtheorem{proposition}[theorem]{Proposition}

\numberwithin{equation}{section}

\def\
vsp{\vspace{0.2cm}}
\def\vspp{\vspace{0.1cm}}

\def\sse #1 {\vsp\ifhmode{\par}\fi\refstepcounter{subsection}
  \noindent {\bf\thesubsection}. {\em #1}.\quad
  \addcontentsline{toc}{subsection}{\protect\numberline{\thesubsection} #1}%
  }

\def\ssb #1 {\vsp\ifhmode{\par}\fi\refstepcounter{subsection}
  \noindent {\bf\thesubsection.} {\bf #1.}\quad
  \addcontentsline{toc}{subsection}{\protect\numberline{\thesubsection} #1}%
  }

\def\ssa #1 {\ifhmode{\par}\fi\refstepcounter{subsection}
  \noindent {\bf\thesubsection.} {\bf #1.}\quad
  \addcontentsline{toc}{subsection}{\protect\numberline{\thesubsection} #1}%
  }

\def\remark #1 {\vsp\vspp\ifhmode{\par}\fi\noindent\noindent {\bf Remark.} {#1}\vsp\vspp\par}
\def\remarks #1 {\vsp\vspp\ifhmode{\par}\fi\noindent\noindent {\bf Remarks.} {#1}\vsp\vspp\par}

\begin{document}


\par
\bigskip
\LARGE
\noindent
{\bf Quantum Spacetime and the Universe at the Big Bang, vanishing interactions and fading degrees of freedom}
\bigskip
\par
\rm
\normalsize


\large
\noindent {\bf Sergio Doplicher$^{1}$}, {\bf Gerardo Morsella$^{2}$}
and {\bf Nicola Pinamonti$^{3}$}\\
\par
\small
\noindent $^1$
Dipartimento di Matematica, Universit\`a di Roma ``La Sapienza'', Piazzale Aldo Moro, 5,
I-00185 Roma, Italy, email dopliche@mat.uniroma1.it.\smallskip

\noindent $^2$
Dipartimento di Matematica, Universit\`a di Roma ``Tor Vergata'', Via della Ricerca Scientifica,
I-00133 Roma, Italy, email morsella@mat.uniroma2.it.\smallskip

\noindent $^3$
Dipartimento di Matematica, Universit\`a di Genova, Via Dodecaneso, 35, 
I-16146 Genova, Italy and 
INFN - Sez. di Genova, Via Dodecaneso, 33 I-16146 Genova, Italy,
email pinamont@dima.unige.it.\smallskip

\normalsize
\par
\medskip

\rm\normalsize
\noindent {\small November 12, 2019}

\rm\normalsize


\par
\bigskip

\noindent
\small
{\bf Abstract}.
As discussed in \cite{BDMP15} Physics suggests that, close to cosmological singularities, the effective Planck length diverges, hence a ``quantum point'' becomes infinitely extended. We argue that, as a consequence, at the origin of times spacetime might reduce effectively to a single point and interactions disappear. This last point is supported by converging evidences in two different approaches to interacting quantum fiedls on Quantum Spacetime: the field operators evaluated at a ``quantum point'' converge to zero, and  so do the lowest order expressions for interacting fields in the Yang Feldman approach, while, at all orders we find convergence to zero of the interacting field operators obtained adapting methods of perturbative Algebraic Quantum Field Theory to Quantum Spacetime, with a novel picture of the effective Lagrangian~\cite{DMP19}. This novel picture mantains the ultraviolet finiteness of the perturbation expansion but allows us to prove also the convergence in the adiabatic limit. It remains an open question whether the $S$ matrix itself converges to unity and whether the limit in which the effective Planck length diverges is a unique initial condition or an unreachable limit, and only different asymptotics matter.

\normalsize

\section{Introduction}
As explained in \cite{BDMP15} Physics suggests that, close to singularities, the effective Planck length diverges. As a consequence, the range of non locality would diverge as well, and this provides a possible solution to the horizon problem \cite{DMP13}. 
However, this means that ``at'' the singularity every point is in instantaneous contact with every other one, that is the Universe, even if infinite and flat at subsequent times, at $t = 0$ was dynamically equivalent to a single point, namely a system with not infinitely many but zero degrees of freedom.
This would solve in a radical way the problem of choices of the ``initial conditions''.

Can we give any mathematical basis to this heuristic argument? One way is the choice of describing interactions of quantum fields on Quantum Space Time (QST) \cite{DFR95} using the ``Quantum Wick Product'' (QWP) \cite{BDFP03}. This choice accounts for the fact that independent points in QST cannot merge to a single point, but at best can be brought near to one another evaluating ``optimally localized states'' on the difference variables - which are quantum variables obeying the same Commutation Relations as the coordinates themselves, and accordingly cannot be set equal to zero. 

Actually, in QST, each of the forms describing distance, area, three-volume, spacetime volume, has a lower bound for the sum of the squares of their components of order one in Planck units (in any reference frame) \cite{BDFP10}.
The use of optimally localized states evaluated on the difference of variables leads also to the formulation of a Quantum Diagonal Map, which composed with the ordinary Wick ordering in the products of field operators evaluated at distinct points in QST, produces the QWP.
Interactions of fields on QST expressed in terms of the QWP, and subject to an adiabatic time cutoff, are free of UV divergences \cite{BDFP03}. But it is also worth recalling that since optimal localization refers to a specific Lorentz frame, the use of the QWP breaks Lorentz covariance.

The scenario described at the beginning is confirmed in a model where matter is simulated by a single scalar massless quantum field, interacting with the gravitational field in a semiclassical picture, where the source of the Einstein Equations is taken to be the expectation value of the quantum energy-momentum tensor of the field - evaluated using the QWP - in a thermal equilibrium state of the field (i.e., a state fulfilling the KMS condition), and only spherically symmetric solutions are considered \cite{DMP13}.

But what happens when the Planck length is replaced by the effective Planck Length, and the latter tends to infinity?
An inspection of the formulas in Proposition 2.3 of \cite{BDFP03}, rewritten for generic values of $\lambda_P$, shows that the Quantum Diagonal Map itself tends to zero when  $\lambda_P$ tends to infinity.
Not surprisingly, since, that map keeps the factors far apart from each other at a distance which becomes infinite in that limit.
This means that, in the above picture of dynamics, the interactions between quantum fields vanish near the singularity: the theory is in some sense {\it asymptotically free} near the Big Bang. The Universe starts as a system with zero degrees of freedom.

More precisely, analyzing the Yang-Feldman equation with interaction term defined through the Quantum Diagonal Map, we show in Sec.~\ref{se:second} below that the interacting field vanishes at lowest order in the limit of large Planck length. Moreover, as discussed in Sec.~\ref{sec:limS}, the same holds, at all orders in perturbation theory, for the interacting observables defined using the adaptation of the perturbative approach to Algebraic Quantum Field Theory (pAQFT) of~\cite{BDF09} developed in~\cite{DMP19}.

The key step in~\cite{DMP19} is the use of a different effective interaction Lagrangian which is equivalent to the one obtained by means of the delocalization kernel discussed in  \cite{BDFP03} in the adiabatic limit, namely after integration of the effective lagrangian density over the full spacetime. This new interaction Lagrangian is obtained by replacing in the classical Lagrangian density the field operators at a point $x$ by fields at a quantum point described by states optimally localized around $x$. 
This makes the associated Feynman rules slightly different from those described by Piacitelli in \cite{piacitelli} c.f. also \cite{Bahns}. The so obtained $S-$matrix turns out to be unitary. The perturbation expansion is still ultraviolet finite, but an important new feature appears: the existence of the adiabatic limit of the vacuum expectation values of interacting observables can be proved in this setting.

The latter result is obtained,   
adapting the analysis of \cite{FredenhagenLindner}, by showing that the spatial adiabatic limit of interacting KMS states at finite temperature can be taken simultaneously with the zero temperature limit, thereby obtaining the ground state of the interacting theory. 
The first crucial observation in order to obtain such result is that, thanks to a remnant of causality in the theory, with a fixed spatial cutoff the temporal cutoff can be taken equal to one in the future of a given time slice. One then observes that, in the limit in which the temporal cutoff is one everywhere, the vacuum expectation values of the interacting observables can also be obtained as the zero temperature limit of the corresponding expectation values in an interacting KMS state. The final step is then to show that the latter limit and the limit in which the spatial adiabatic cutoff is removed can be taken together. The resulting state is therefore invariant under spacetime translations, and can then be interpreted as the vacuum of the interacting theory. This shows in particular that
in the obtained theory there is no ultraviolet-infrared mixing.

The conclusions on the limit when the effective Planck length diverges are better understood if we consider the generalization of the notion of ``field at a point'' in Quantum Spacetime. 
Points have to be replaced by suitable pure states on the Algebra of Quantum Spacetime, where the uncertainties in the four coordinates are simultaneously as small as possible.
These states are the ``optimal localized states'', mentioned above. Evaluating a field operator - as a function of the quantum coordinates, affiliated to the tensor product of the algebra of Quantum Spacetime and that of field operators - on such a state in the first factor, tensored with the identity map in the second, produces the desired generalization of the notion of ``field at a point'' in Quantum Spacetime \cite{DFR95}.

If in this expression we replace the Planck length by the effective Planck length, in the limit where the latter tends to infinity, the so constructed  ``field at a point'' becomes translation invariant and tends to a multiple of the identity, generically to zero.

This argument is in agreement with the picture that the interactions between Quantum Fields vanish near the singularity; but note that it does not rely on the adoption of the Quantum Wick Product.

However, the initial singularity might well be an unreachable limit: the possible different ways to approach that limit might well replace, with a manyfold of possibilities, the different choices of initial conditions, the existence of a unique limit being replaced by a chaos of several singular limiting points.

At the opposite extreme, it might turn out that the vanishing of interactions at the limit already produces the gradual vanishing of the differences between different states as the singularity is approached. And the last mentioned scenario might well be valid for special classes of interactions, which are asymptotically free near singularities.

These questions are closely related to the problem, which is quite open,  of the transition from the possibly single state of the system with zero degrees of freedom representing the universe at $t = 0$ to the states of the system with infinitely many degrees of freedom representing the flat universe at subsequent times.

Quantum Spacetime thus seems to solve the Horizon Problem and might well solve that of the initial conditions. Similarly, the extension of size equal to the effective Planck length of the quantum oscillations could well justify the emergence of seeds of the formation of galaxies in the early instants, when the effective Planck length was extremely large, just from quantum fluctuations of the vacuum.

So, many consequences of the inflationary hypothesis seem to be within reach of explanations based on the Quantum nature of Spacetime; the quantum structure in the small seems to be at the root of various features of the universe in the large. 

Current researches are concerned with the question whether the inflationary hypothesis itself, as well as a specific form of an ``effective inflationary potential'', can be deduced from QST.

\section{Quantum fields on QST in the $\la_P \to \infty$ limit}\label{se:second}
In the approach of~\cite{DFR95}, Quantum Spacetime is the C*-algebra $\cE$ describing the realizations of the relations (in generic units)
\[
e^{i k_\mu q^\mu} e^{i h_\mu q^\mu} = e^{-\frac {i\la_P^2} 2 k_\mu h_\nu Q^{\mu \nu}} e^{i (k_\mu + h_\mu) q^\mu}, \qquad h,k \in \bR^4,
\] 
with $q^\mu, Q^{\mu \nu}$, $\mu,\nu = 0,\dots,3$, selfadjoint operators such that $Q^{\mu \nu}$ equals (the closure of) the commutator $-i\la_P^{-2}[q^\mu, q^\nu]$, strongly commutes with $q^\la$, and satisfies
\begin{equation}\label{eq:qc}
Q^{\mu \nu} Q_{\mu \nu}  = 0, \qquad
\qquad \frac 1 {4} (Q_{\mu \nu}(*Q)^{\mu \nu})^2
= \id.
\end{equation}
These ``quantum conditions'' implement physically motivated spacetime uncertainty relations~\cite{DFR95}. It turns out that $\cE \simeq C_0(\Sigma, \cK)$, with $\cK$ the compact operators on the infinite dimensional separable Hilbert space and $\Sigma \simeq T\bS^2 \times \{-1,1\}$ the joint spectrum of the $Q^{\mu\nu}$'s in covariant representations, i.e., the manifold of real antisymmetric $4\times 4$ matrices satisfying~\eqref{eq:qc}. 

Optimally localized states on $\cE$ minimize the sums of the squares of the uncertainties of the $q_\mu$'s and are therefore the best approximations of points. As such, these states are parametrized by a localization center $x \in \bR^4$, and by a probability measure on $\Sigma$, which is automatically concentrated on the basis $\Sigma_1 \simeq \bS^2 \times \{-1,1\}$.  A state $\om_x$ optimally localized at $x$ is such that (its normal extension to the multiplier algebra of $\cE$ satisfies)
\begin{equation}\label{eq:optimal-localization}
\om_x(e^{i k_\mu q^\mu}) = e^{i k \cdot x}e^{-\frac{\la^2_P}{2} |k|^2}, \qquad k \in \bR^4,
\end{equation}
where $k \cdot x = k_\mu x^\mu$ is the usual Minkwoski scalar product, and $|k|^2 = \sum_{\mu=0}^3 k_\mu^2$ is the Euclidean norm square.

As mentioned in the Introduction, the quantum diagonal map introduced in~\cite{BDFP03} is a generalization to QST of the evaluation at coinciding points of a function of several points of $\bR^4$. It is therefore a map $E^{(n)} : \cE^{\otimes_Z^n} \to \cE_1$, with $Z=C_b(\Sigma)$ the center of the multiplier algebra of $\cE$, $\otimes_Z$ the $Z$-moduli tensor product and $\cE_1 \simeq C(\Sigma_1,\cK)$, the restriction of $\cE$ to $\Sigma_1$. On elements of $\cE^{\otimes_Z^n}$ of the form
\[
f(q_1,\dots,q_n) = \int_{\bR^{4n}} dk_1\dots dk_n\, \hat f(k_1,\dots,k_n) e^{ik_1 \cdot q_1}\dots e^{ik_n\cdot q_n},
\]
where $\hat f \in L^1(\bR^{4n})$ and $f(x_1,\dots,x_n) = \int_{\bR^{4n}} dk_1\dots dk_n\, \hat f(k_1,\dots,k_n) e^{ik_1 \cdot x_1}\dots e^{ik_n\cdot x_n}$ its (inverse) Fourier transform, the quantum diagonal map can be computed as
\[
E^{(n)}(f(q_1,\dots,q_n)) = \int_{\bR^{4n}} dk_1\dots dk_n\, \hat f(k_1,\dots,k_n) \hat r_n(k_1,\dots,k_n) e^{i(k_1+\dots+k_n) \cdot q},
\]
where
\[
\hat r_n(k_1,\dots,k_n) = e^{-\frac{\la_P^2}{2} \sum_{j=1}^n \left|k_j - \frac 1 n \sum_{i=1}^n k_i\right|^2}.
\]

The free scalar field can be defined on QST in analogy to Wigner-Weyl calculus:
\begin{equation}\label{eq:field}
\phi(q) := \int_{\bR^4} dk\, \hat \phi(k) \otimes e^{i k\cdot q}.
\end{equation}
The above formula has to be interpreted, more precisely, as an affine map from (a suitable *weakly dense subset of the) states on $\cE$ into the *-algebra $\cF$ of polynomials in the field operators on (commutative) Minkowski space time, given by
\[
\om \mapsto (\mathrm{id}\otimes \om)(\phi(q)) = \phi(f_\om) = \int_{\bR^4} dx\,\phi(x) f_\om(x),
\] 
where the right hand side is the operator on the bosonic Fock space obtained by smearing the ordinary free scalar field with the test function $f_\om(x) := \int_{\bR^4} dk\,\om(e^{ik\cdot q}) e^{-ik\cdot x}$, assuming that $k \mapsto \om(e^{i k\cdot q})$ belongs to the Schwarz space $\cS(\bR^4)$. 

The $n$-th quantum Wick power of $\phi$ is then obtained by composing the ordinary Wick product with the quantum diagonal map
\[\begin{split}
:\phi^n:_Q(q) &:= E^{(n)}(:\phi(q_1)\dots \phi(q_n):) \\
&= \int_{\bR^{4n}} dk_1\dots dk_n\,  \hat r_n(k_1,\dots,k_n):\hat \phi(k_1)\dots\hat\phi(k_n):\otimes \,e^{i(k_1+\dots+k_n) \cdot q}.
\end{split}
\]

The use of the ordinary Wick product that, when the Quantum Diagonal Map  is adopted on Quantum Spacetime, subtracts finite terms, is dictated by the requirement that the ordinary theory ought to emerge in the limit where the Planck length is neglected.

The above framework describes quantum fields on a quantum version of Minkowski spacetime. In order to discuss in a quantitative way the influence of the quantum nature of spacetime on the early cosmological evolution, it seems desirable to generalize this approach to curved backgrounds. Although some progress in this direction has been made~\cite{TV14, MT}, we still lack a complete theory of quantum fields on a curved quantum spacetime.
  
 However, some indications on the possible modifications to standard cosmology induced by the energy density of quantum fields propagating on QST have been obtained through a semiclassical approach in~\cite{DMP13}. 
 
 More precisely, a generalization of the analysis in~\cite{DFR95} of the operational limitations of localizability of events shows that also on a spherically symmetric, curved spacetime there is a minimal spherical localization distance of the order of the Planck length. This implies in particular that on a flat Friedmann-Robertson-Walker spacetime with metric
\[
ds^2 = -dt^2 +a(t)^2 \sum_{j=1}^3 dx_j^2
\] 
the \emph{effective} Planck length, i.e. the minimal localization distance measured in comoving coordinates, behaves like $1/a(t)$ for small scale factor $a(t)$ (and therefore diverges at the Big Bang, thus confirming the previous heuristic discussion of~\cite{Dop01}), or, equivalently, is independent of time if measured in terms of proper lengths. 

Combining this observation with the fact that flat FRW is conformally equivalent to Min\-kow\-ski spacetime, one is led to conjecture that the expectation value in a conformal KMS state of the energy density of a free massless conformally coupled scalar field on a quantum version of FRW spacetime can be obtained from the corresponding expectation value in a KMS state on quantum Minkowski spacetime by the same change $\beta \to \beta a(t)$ which works on the corresponding classical spacetimes. Thus, if one defines the energy density through the quantum Wick product, one obtains the \emph{ansatz}
\[
\rho_\beta(t) = \frac 1 {2\pi} \int_0^{+\infty} dk \,k^3 \frac{e^{-\la_P^2 k^2}}{e^{\beta a(t)k}-1}
\]
for the expectation value of the energy density on quantum FRW. In particular this expression has the asymptotic behavior $\rho_\beta(t) \sim a(t)^{-1}$ for small $a(t)$ and this entails that, when $\rho_\beta$ is coupled to the metric through the Friedmann equation, the initial singularity of the resulting cosmological evolution is removed to lightlike past infinity, thus avoiding the horizon problem of standard cosmology. It is worth stressing that this result is only a consequence of the appearance of a new length scale dictated by the quantum nature of spacetime, not of the choice of a particular dynamics of the quantum fields considered as, e.g., in inflationary models.
  
In order to discuss a more realistic situation, it is of course necessary to define an interacting field on QST. To this end, one possibility  is to employ the Yang-Feldman approach, consisting in solving iteratively the field equation, where the interaction term is defined via the quantum Wick product (unlike~\cite{BDFP05}, where the ordinary product in $\cE$ is used instead). To be more specific, we consider the equation
\begin{equation}\label{eq:YF}
(-\Box+m^2) \phi(q) = -g E^{(3)}\big(\phi(q_1)\phi(q_2)\phi(q_3)\big),
\end{equation}
with $\partial_\mu \phi(q) := \frac{\partial}{\partial a^\mu} \phi(q+a \id)|_{a=0}$, and look for a (formal) power series solution
\[
\phi(q) = \sum_{k=0}^{+\infty} g^k \phi_k(q).
\]
Inserting this series into~\eqref{eq:YF} and equating the coefficients of the corresponding powers of $g$, one finds that, obviously, $\phi_0(q)$ is the free scalar field~\eqref{eq:field}, while $\phi_1(q)$ would have to satisfy $(\Box +m^2)\phi_1(q) = - E^{(3)}(\phi_0(q_1)\phi_0(q_2)\phi_0(q_3))$. This has the problem that the right hand side contains the expression $\hat\phi_0(k_1)\hat \phi_0(k_2) \hat\phi_0(k_3)$, which is not even a quadratic form on Fock space. In order to make it well-defined, it is then natural to renormalize it as $:\phi_0^3:_Q(q)$. Therefore, identifying $\phi_0$ with the incoming field, 
\[
\phi_1(q) = -\int_{\bR^4} dx\,\De_R(x) :\phi_0^3:_Q(q+x),
\]
with $\De_R \in \cS(\bR^4)'$ the retarded propagator for the Klein-Gordon equation.
 
As anticipated in the Introduction, the interacting field on QST vanishes in the limit $\la_P \to \infty$, at least at first order in perturbation theory.

\begin{proposition}
For each Fock space vector $\Psi$ such that $k \mapsto \langle \Psi,\hat \phi_0(k)\Psi\rangle$, $(k_1,k_2,k_3) \mapsto \langle \Psi, : \hat\phi_0(k_1)\hat \phi_0(k_2) \hat\phi_0(k_3):\Psi\rangle$ are Schwartz functions, the limit 
\[
\lim_{\la_P \to +\infty} \langle \Psi,(\mathrm{id}\otimes\om_x)(\phi(q)) \Psi\rangle
\]
vanishes, at least at the first order in the perturbative expansion.

\end{proposition}

\begin{proof}
By the above formulas, one has $\langle \Psi,(\mathrm{id}\otimes\om_x)(\phi_1(q)) \Psi\rangle = -(2\pi)^4\int_{\bR^4} dk\,\hat\De_R(k) e^{-i k \cdot x} f_{\la_P}(k)$, where
\[\begin{split}
f_{\la_P}(k) &:= e^{-\frac{\la^2_P}{2} |k|^2}
 \int_{\bR^{12}} dk_1dk_2 dk_3\, \de\left(k - \sum_{j=1}^3 k_j\right)\langle \Psi, : \hat\phi_0(k_1)\hat \phi_0(k_2) \hat\phi_0(k_3):\Psi\rangle \hat r_3(k_1,k_2,k_3)\\
& = e^{-\frac{\la^2_P}{2} |k|^2} \int_{\bR^8} d\ka_1 d\ka_2\, \psi\Big(\ka_1+\frac k 3,\ka_2+\frac k 3,\frac k 3-\ka_1-\ka_2\Big) e^{- \frac{\la_P^2}{2}(|\ka_1|^2+|\ka_2|^2+|\ka_1+\ka_2|^2)},
\end{split}\]
having defined $\psi(k_1,k_2,k_3) := \langle \Psi, : \hat\phi_0(k_1)\hat \phi_0(k_2) \hat\phi_0(k_3):\Psi\rangle$. Therefore, in order to show that $\langle \Psi,(\mathrm{id}\otimes\om_x)(\phi_1(q)) \Psi\rangle \to 0$ as $\la_P \to +\infty$, it is sufficient to verify that $f_{\la_P} \to 0$ in $\cS(\bR^4)$. To this end, we observe that the fact that $\psi \in \cS(\bR^{12})$ entails that for all $m \in \bN$ there exists $C_m > 0$ such that
\[
\sup_{k \in \bR^4} (1+|k|^2)^m |\partial^\alp f_{\la_P}(k)| \leq C_m \int_{\bR^8} d\ka_1 d\ka_2 \sup_{k \in \bR^4} \frac{(1+|k|^2)^m}{(1+|\ka_1+k/3|^2)^{m}} e^{- \frac{\la_P^2}{2}(|\ka_1|^2+|\ka_2|^2+|\ka_1+\ka_2|^2)},
\]
and since, clearly,
\[
\sup_{k \in \bR^4} \frac{(1+|k|^2)^m}{(1+|\ka_1+k/3|^2)^{m}} = \sup_{k \in \bR^4} \frac{(1+9|k-\ka_1|^2)^m}{(1+|k|^2)^{m}} \leq C(1+|\ka_1|^2)^m
\]
for a suitable $C > 0$, we obtain
\[
\sup_{k \in \bR^4} (1+|k|^2)^m |\partial^\alp f_{\la_P}(k)| \leq C_m C\int_{\bR^8} d\ka_1 d\ka_2\, (1+|\ka_1|^2)^m e^{- \frac{\la_P^2}{2}(|\ka_1|^2+|\ka_2|^2+|\ka_1+\ka_2|^2)}.
\]
The vanishing of the right hand side of the latter formula for $\la_P \to \infty$ is then an immediate consequence of the dominated convergence theorem.

Moreover, again the dominated convergence theorem guarantees that
\[
\langle \Psi,(\mathrm{id}\otimes\om_x)(\phi_0(q)) \Psi\rangle  = \int_{\bR^4}dk \langle \Psi, \hat\phi_0(k)\Psi\rangle e^{i k\cdot x} e^{-\frac{\la_P^2}{2}|k|^2}
\]
vanishes too as $\la_P \to +\infty$, and we get the desired statement.
\end{proof}

We note that the above argument can also be employed to show that  $\langle \Psi,(\mathrm{id}\otimes\om)(\phi_1(q)) \Psi\rangle \to 0$ for an arbitrary state $\om$ on $\cE$.

\section{Limits for large $\lambda_P$ at all orders in perturbation theory}\label{sec:limS}
In this Section we shall discuss the form of the correlation functions of the theory in the limit for large $\lambda_P$. 
Following \cite{DFR95} and the discussion presented in section \ref{se:second}, this limit correspond to an early time limit in a cosmological scenario. A convenient framework in which this limit can be controlled at all orders in perturbation theory is the one deveolped in~\cite{DMP19} adapting to QFT on QST the framework of pAQFT. To make the paper reasonably self-contained, we summarize below the main definitions and results of~\cite{DMP19}, to which we refer the interested reader for more details and proofs. To simplify the notation, from now on we denote the Planck length just by $\lambda$.

The starting observation is that the spacetime integral of the effective interaction Lagrangian of~\cite{BDFP03} coincides with the adiabatic (i.e., $g \to 1$) limit of 
\begin{equation}\label{eq:eff-adiab}
{V}_{\text{eff}} (\phi)= 
\left( \frac{\sqrt{2\pi}}{\lambda}\right)^{4n}
\int_{\bR^4}  dx  \; g(x) \prod_j \left(\int_{\bR^4}  dy_j \; e^{-\frac{ \langle y_j-x\rangle^2}{2 \lambda^2 }}
\phi(y_j)\right),
\end{equation}
which is what is obtained by replacing in the classical Lagrangian density fields at a point $x$ with fields at a quantum point described by the state optimally localized at $x$. As a consequence, the algebra of observables in the effective theory with interaction~\eqref{eq:eff-adiab} can be mapped with an appropriate bijection onto the algebra generated by the operators $(\mathrm{id}\otimes \omega_x)(\phi(q))$ with the local interaction
\begin{equation}\label{eq:V}
V_g(\phi) = \int_{\bR^4} g(x) \phi^n(x)   dx.
\end{equation}
To describe such algebra more explicitly, consider the set $\cA$ of functionals $A: C^\infty(\bR^4;\bR) \cap \cS'(\bR^4) \to \bC$ which are finte sums of functionals of the form
\begin{equation}\label{eq:A}
A(\phi) = \int_{\bR^{4\ell}} dx_1\dots dx_\ell \,f(x_1,\dots,x_\ell) \phi(x_1)^{n_1}\dots \phi(x_\ell)^{n_\ell},
\end{equation}
where $f : \bR^{4\ell} \to \bC$ is a continuous and compactly supported function. Such a functional is said to be localized in an open set $O \subset \bR^4$ if  $\supp f \subset O^\ell$. We denote by $\cA(O)$ the set of all such functionals. Define now the product of functionals $A, B \in \cA$ by
\[
A\star_\la B =  \mathcal{M} \circ e^{ \Gamma_{\Delta_{+,\la}} }  (A\otimes B)
\]
where $\mathcal{M}(A\otimes B)(\phi) = A(\phi)B(\phi)$ is the pointwise multiplication
and 
\[
\Gamma_{\Delta_{+,\la}} = \int_{\bR^8} dxdy \,\Delta_{+,\la}(x-y)\frac{\delta}{\delta \phi(x)}\otimes \frac{\delta}{\delta \phi(y)},
\]
with $\Delta_{+,\la}$ the two-point function of the field at a quantum point on QST:
\begin{equation}\label{eq:delta+lambda}
\Delta_{+,\la}(x) := \langle \Omega, (\mathrm{id}\otimes \omega_x)(\phi(q)) (\mathrm{id}\otimes \omega_0)(\phi(q)) \Omega\rangle = \int_{\bR^8} dydz\, G_\la(x-y)\Delta_+(y-z) G_\la(z), 
\end{equation}
and
\[
G_\la(x) := \frac {e^{-\frac{|x|^2}{2\la^2}}}{(\sqrt{2\pi}\la)^4}.
\]
This entails in particular that $\Delta_{+,\lambda}$ is a smooth and bounded function, and therefore $\cA$ becomes a $*$-algebra, the \emph{algebra of free observables}, with the product $\star_\la$ and the involution $A^*(\phi) := \overline{A(\phi)}$ (complex conjugation). The time-ordered  product associated to the product $\star_\la$ is then given by
 \[
A\cdot_{T_\la} B =  \mathcal{M} \circ e^{ \Gamma_{\Delta_{F,\la}} }  (A\otimes B),
\]
with the modified Feynmann propagator
\begin{equation}\label{eq:feynman}
{\Delta}_{F,\lambda}(x) := \theta(x^0){\Delta}_{+,\lambda}(x)+\theta(-x^0){\Delta}_{+,\lambda}(-x), 
\end{equation}
which is a continuous and bounded function, with Fourier transform
\begin{equation}\label{eq:modFeyn}
\hat{{\Delta}}_{F,\lambda}(p) =
\frac{-i}{(2\pi)^4}
 \frac{1}{ p^2 +m^2 -i\epsilon}
 e^{-\lambda^2 (2|\mathbf{p}|^2+m^2)}.
\end{equation}

In the following proposition we analyze the behavior of the above defined products for large Planck length.

\begin{proposition}\label{pr:factor-limit}
 For every $A,B\in\mathcal{A}$ the following limits hold
\[
\lim_{\lambda\to\infty}A\star_\lambda B = A\cdot B,  
\qquad
\lim_{\lambda\to\infty}A \cdot_{T_\lambda} B = A \cdot B,
\]
where $\cdot$ denotes the pointwise product.
\end{proposition}
\begin{proof}
In the limit $\lambda\to \infty$ both $\Delta_{+,\lambda}$ and $\Delta_{F,\lambda}$ vanish as can be seen directly from  \eqref{eq:delta+lambda} and from \eqref{eq:modFeyn}. Hence, in the limit $\lambda \to \infty$ both $\Gamma_{\Delta_{+,\lambda}}$ and $\Gamma_{\Delta_{F,\lambda}}$ vanish, and thus
\[
\lim_{\lambda\to\infty}A \cdot_{T_\lambda} B = 
\lim_{\lambda\to\infty} \mathcal{M} \circ e^{ \Gamma_{i\Delta_{F,\lambda}}} (A\otimes B) = A\cdot B.
\]
The same holds for the $\star_{\lambda}$ product.
\end{proof}

The $\cdot_{T_\la}$ product is then used in~\cite{DMP19} to define the (infrared cutoff) $S$ matrix corresponding to the interaction~\eqref{eq:V} as
\begin{align*}
S(V_g) &= \sum_{k=0}^{+\infty} \frac{(-i)^k}{k!}V_g\cdot_{T_\la} V_g \cdot_{T_\la} \dots \cdot_{T_\la} V_g  \\
&=  \sum_{k=0}^{+\infty} \frac{(-i)^k}{k!} \int_{\bR^{4k}} dx_1\dots dx_k\,g(x_1)\dots g(x_k) \big(\phi(x_1)^n \cdot_{T_{\la}}\dots \cdot_{T_{\la}}\phi(x_k)^n\big),
\end{align*}
which is a formal power series of functionals in $\cA$. The resulting perturbation expansion can be obtained from the corresponding unrenormalized one on classical spacetime simply by substituting the ordinary Feynmann propagator with the modified one~\eqref{eq:modFeyn}. It is a remarkable fact that the above $S$ matrix is automatically ultraviolet finite, without the need of renormalization, thanks to the boundedness of the propagator $\Delta_{F,\la}$, and that it is unitary, i.e., it satisfies
\[
S(V_g) \star_\la S(V_g)^* = 1 =S(V_g)^*\star_\la S(V_g)
\]
to all orders in perturbation theory. But of course this $S$ matrix can have a physical meaning only after an approriate finite renormalization, cf.\ the discussion of this problem in~\cite[Sec.\ 6]{DMP19}.

The perturbative expansion of \emph{interacting observables} is then defined by means of the Bogoliubov map:
\[
R_{V_g}(A) := S(V_g)^{-1} \star_\la (S(V_g) \cdot_{T_\la} A), \qquad A \in \cA.
\]

Proposition~\ref{pr:factor-limit} implies in particular that the $S$ matrix in $\mathcal{A}$ is particularly simple in the $\la \to \infty$ limit. Actually the following proposition holds.

\begin{proposition}\label{pr:S-R-limit}
Let $V_g$ be as in~\eqref{eq:V} and $A \in \cA$. Then under the limits of large $\lambda$
\[
\lim_{\lambda\to\infty} S(V_g) = \exp{iV_g}, \qquad \lim_{\lambda\to\infty} R_{V_g}(A) =  A. 
\]
where $\exp{iV_g}$ is the exponential constructed with the pointwise product and where both limits are taken in the sense of perturbation theory.
\end{proposition}
\begin{proof}
In order to check the validity of both limits in the sense of perturbation theory, it is enough to expand $S(\epsilon V_g)$ and $R_{\epsilon V_g}(A)$ as a power series in $\epsilon$ and to consider the contribution at a generic fixed order $n$.
At that order, $S(V_g)$ is a time ordered product of a finite number of copies of $V_g$. Proposition \ref{pr:factor-limit}
implies that under that limit the time ordered product converges to the pointwise product. Hence we have proved that the first limit holds. To prove the second limit we have just to expand $R_{V_g}(A)$ as a sum of mixed (time ordered or $\star_\lambda$) products  of $n$ copies of  $V_g$ and $A$. Again in the limit $\lambda\to\infty$ all these products converge to the pointwise product, hence we have the thesis.
\end{proof}

We now want to show that a similar result also holds in the adiabatic limit discussed in \cite{DMP19}. To this end, we first briefly summarize here its construction. We first introduce the free vacuum and KMS states $\omega$ and $\omega_{\beta, \la}$ on $\cA$:
\[
\omega(A) := A(0), \qquad \omega_{\beta,\la}(A) := \omega(e^{\frac 1 2 \int dxdy\,[\Delta_{\beta,\la}(x-y)-\Delta_{+\,\la}(x-y)]\frac{\delta^2}{\delta \phi(x) \delta \phi(y)}}A),\quad A \in \cA,
\]
where, in analogy to the 2-point function, 
\[
\Delta_{\beta,\la}(x) := \int_{\bR^8} dydz\, G_\la(x-y)\Delta_\beta(y-z) G_\la(z)
\]
and $\Delta_\beta$ is the usual thermal free 2-point function. It is clear from these formulas that $\omega_{\beta, \la}$ converges to $\omega$ as $\beta \to +\infty$. One then introduces the free time evolution $\alpha_t$, $t \in \bR$, on $\cA$, defined by
\[
\alpha_t(A)(\phi) := A(\phi_t),\quad \phi_t(x_0,\mathbf{x}) := \phi(x_0+t,\mathbf{x}).
\]
The algebra $\cA$ can also be equipped with the interacting time evolution $\alpha^V_t$, $t \in \bR$, which is intertwined with the free one by the Bogoliubov map $R_{V_g}$:
\[
\alpha^V_t(R_{V_g}(A)) = R_{V_g}(\alpha_t(A)), \qquad t \in \bR.
\]
We now observe that thanks to a remnant, in our nonlocal setting, of the casuality properties of the $S$ matrix, the interacting observables $R_{V_g}(A)$ do not depend on the behavior of the infrared cutoff function $g$ in the future of the localization region of $A$. This implies that, if $A$ is localized in the future of the hyperplane $t = -\epsilon$, we can choose the infrared cutoff in $R_{V_g}(A)$ as $g(t,\mathbf{x}) = \chi(t) h(\mathbf{x})$, where $h$ is a compactly supported function on $\bR^3$ and $\chi$ is  such that $\chi(t) = 0$ for $t \leq -2\epsilon$ and $\chi(t) = 1$ for $t \geq -\epsilon$. In this case, we write $V_{\chi,h} = V_g$. Then, following~\cite{FredenhagenLindner}, it is possible to show that, for $A$ localized as above, the free and interacting time evolution are related by
\[
\alpha^V_t(R_{V_{\chi,h}}(A)) = U(t) \star_\la \alpha_t(R_{V_{\chi,h}}(A)) \star_\la U(t)^*, \qquad t \in \bR,
\]
where $U(t) = S(V_{\chi,h})^{-1}\star_\la S(\alpha_t V_{\chi,h})$, $t \in \bR$, are unitaries in $\cA$ which form a cocyle with respect to the free evolution: $U(t+s) = U(t) \star_\la \alpha_tU(s)$. Using such cocycle, one can then construct the interacting KMS states $\omega^h_{\beta, \la}$, defined formally as
\[
\omega^h_{\beta,\la}(A) := \frac{\omega_{\beta,\la}(U(i\beta/2)^*\star_\la A\star_\la U(i\beta/2))}{\omega_{\beta,\la}(U(i\beta/2)^*\star_\la U(i\beta/2))},
\]
i.e., more precisely, as the analytic continuation for $t = -i\beta/2$ of the function $t \mapsto \omega_{\beta,\la}(U(-t)^*\star_\la A\star_\la U(-t))$. Finally, it is shown in~\cite{DMP19} that there exists, order by order in perturbation theory, the limit
\[
\tilde \omega_\la(A) =\lim_{h \to 1} \lim_{t \to +\infty} \omega(\alpha_t^VR_{V_{\chi,h}}(A)) = \lim_{\substack{\beta \to \infty\\h \to 1}} \omega^h_{\beta,\la}(R_{V_{\chi,h}}(A)),
\]
where the order of the limits in the right hand side is immaterial, and that the resulting state $\tilde \omega_\la$ is invariant under spacetime translation, and can therefore be interpreted as the vacuum of the interacting theory.

The next proposition then shows that in the limit of large $\la$,  the above defined interacting vacuum reduces to the free one.
\begin{proposition}
For $A \in \cA(O)$ with $O \subset (-\epsilon,\epsilon)\times \bR^3$ bounded, there holds
\[
\lim_{\la \to \infty} \tilde\omega_\la(A) = \omega(A)
\]
to all orders in perturbation theory.
\end{proposition}
\begin{proof}
It follows from Prop.~\ref{pr:S-R-limit} and the above formulas that
\[
\lim_{\la \to +\infty} U(t) = e^{-iV_{\chi,h}} e^{i V_{\chi,h}} = 1, \qquad \lim_{\la \to +\infty} \omega_{\beta,\la}(A) = \omega(A).
\]
This clearly implies, using again Prop.~\ref{pr:S-R-limit},
\[
\lim_{\la \to +\infty} \omega_{\beta,\la}(U(-t)^* \star_\la R_{V_{\chi,h}}(A)\star_\la U(-t)) = \omega(A),
\]
and therefore $\lim_{\la \to +\infty} \omega_{\beta,\la}^h(R_{V_{\chi,h}}(A)) = \omega(A)$ for fixed $h$ and $\beta$. 
It is then easy to verify that the estimates of~\cite[Prop.\ B.1]{DMP19} hold uniformly for $\la > 1$, and one sees, by inspection of the proof of~\cite[Thm.\ 5.5]{DMP19}, that this implies that the $\la \to +\infty$ limit can be interchanged with the $\beta \to +\infty$, $h \to 1$ one, thus concluding the proof.
\end{proof}

The discussion presented in this section indicates that at early time, in a cosmological quantum spacetime every interacting quantum field theory should converge to a free theory.

\section{Outlook}

The physical model of Quantum Spacetime was first proposed \cite{DFR95} as a compelling joint consequence of Quantum Mechanics and General Relativity, in contrast with toy mathematical models, and with the model which was early proposed by Snyder \cite{Snyder} as a possible way out of divergences (but there was no reference to Gravitation there, and the model was soon forgotten with the advent of Renormalisation). Motivations partly similar to those in  \cite{DFR95} were present in \cite{Maggiore} and \cite{Madore}.

But the Quantum Conditions, apt to implement the Uncertainty Relations that Physics dictates, ought to be derived dynamically in a consistent quantum theory of gravitation, which is still missing. Only in such a theory could one deduce rigorously features occurring in the cosmology of the early Universe. The present state of the art allows one only to insert in the formalism nowadays available some features that Quantum Spacetime strongly indicates as characteristics to be expected: the use of the Quantum Wick Product to describe interactions, and the replacement of the Planck length by the effective Planck Length, which would diverge approaching the Big Bang. 

In order to fully confirm the scenario presented in this paper, we still miss a proof that the perturbation expansion of the $S$ matrix itself will converge term by term to unity as $\lambda$ tends to infinity.

While exposed to be challenged by a future complete theory, this approach leads to a consistent picture, which makes plausible a solution of the Horizon Problem without introducing an inflationary  potential \cite{DMP13}, and provides the evidence discussed here that degrees of freedom and interactions might fade away at the beginning of times. To be compared with a future theory where spacetime commutation relations ought to be part of the equation of motion, fusing Algebra with Dynamics \cite{Dop01, BDMP15}.

\end{document}